\documentclass[11pt]{article}

\usepackage{color}
\definecolor{ForestGreen}{rgb}{0.1333,0.5451,0.1333}
\usepackage[colorlinks=true,linkcolor=blue,citecolor=ForestGreen]{hyperref}
\usepackage{amsmath, amssymb, amsthm}
\usepackage{mathtools}
\usepackage[noabbrev,capitalize,nameinlink]{cleveref}
\crefname{equation}{}{}
\usepackage{fullpage}
\usepackage{graphics}
\usepackage{pifont}
\usepackage{tikz}
\usepackage{bbm}
\usepackage[T1]{fontenc}

\usetikzlibrary{arrows.meta}

\usepackage{environ}
\usepackage{framed}
\usepackage{url}
\usepackage[lined,boxed,ruled,norelsize,algo2e,linesnumbered,noend]{algorithm2e}
\usepackage[noend]{algpseudocode}
\usepackage[labelfont=bf]{caption}
\usepackage{cite}
\usepackage{framed}
\usepackage[framemethod=tikz]{mdframed}
\usepackage{appendix}
\usepackage{graphicx}
\usepackage[textsize=tiny]{todonotes}
\usepackage{tcolorbox}
\allowdisplaybreaks[1]

\newcommand\remove[1]{}

\newtheorem{theorem}{Theorem}
\newtheorem{lemma}{Lemma}[section]
\newtheorem*{lemma*}{Lemma}
\newtheorem{corollary}[lemma]{Corollary}
\newtheorem*{corollary*}{Corollary}

\theoremstyle{definition}

\newtheorem*{theorem*}{Theorem}
\newtheorem{definition}[lemma]{Definition}

\newtheorem*{rem*}{Remark}

\newcommand\F{\mathbb{F}}

\newcommand{\eps}{\varepsilon}
\renewcommand{\O}{\widetilde{O}}

\newcommand{\bs}{\backslash}

\newcommand{\assign}{\leftarrow}

\renewcommand{\forall}{\mathrm{\text{ for all }}}

\newcommand{\T}{\mathcal{T}}

\newcommand{\mincut}{\mathrm{mincut}}

\newcommand{\new}{\mathrm{new}}
\newcommand{\I}{\mathcal{I}}
\newcommand{\M}{\mathcal{M}}
\newcommand{\lovasz}{Lov\'{a}sz}
\newcommand{\J}{\mathcal{J}}
\renewcommand{\sc}{\mathrm{sc}}
\newcommand{\Existence}{\textsc{Existence}}
\newcommand{\PolyTimeNetwork}{\textsc{PolyTimeNetwork}}
\renewcommand{\split}{\mathrm{split}}

\newif\ifrandom
\randomtrue

\newcommand{\defeq}{\stackrel{\mathrm{\scriptscriptstyle def}}{=}}

\newcommand{\poly}{{\mathrm{poly}}}

\newcommand{\todolater}[1]{}
\newcommand{\rank}{\mathrm{rank}}
\newcommand{\tw}{\mathrm{tw}}

\crefname{algocf}{Algorithm}{Algorithms}

\author{Yang P. Liu\\Stanford University \\
\texttt{yangpliu@stanford.edu}
\thanks{Research supported by the Department of Defense (DoD) through the National Defense
Science and Engineering Graduate Fellowship (NDSEG) Program.}}

\begin{document}

\title{Vertex Sparsification for Edge Connectivity \\ in Polynomial Time}

\begin{titlepage}
\clearpage\maketitle
\thispagestyle{empty}

\begin{abstract}
An important open question in the area of vertex sparsification is whether $(1+\eps)$-approximate cut-preserving vertex sparsifiers with size close to the number of terminals exist. The work \cite{CDL21} (SODA 2021) introduced a relaxation called \emph{connectivity-$c$ mimicking networks}, which asks to construct a vertex sparsifier which preserves connectivity among $k$ terminals exactly up to the value of $c$, and showed applications to dynamic connectivity data structures and survivable network design. We show that connectivity-$c$ mimicking networks with $\O(kc^3)$ edges exist and can be constructed in polynomial time in $n$ and $c$, improving over the results of \cite{CDL21} for any $c \ge \log n$, whose runtimes depended exponentially on $c$.
\end{abstract}

\end{titlepage}

\newpage

\section{Introduction}
\label{sec:intro}

Sparsification is the fundamental concept of reducing the size of a large graph $G$ while still maintaining essential properties of the graph. Important examples include spanners \cite{Chew89}, which approximately preserve distances up to a multiplicative factor, or cut and spectral sparsifiers \cite{BK96,ST04}, which preserve the cuts and Laplacian spectrum of the graph up to a $(1+\eps)$ factor. These \emph{edge sparsifiers} allow one to reduce problems on dense graphs to sparse graphs at the cost of an approximation factor resulting from the sparsification. On the other hand, several methods such as elimination-based Laplacian solvers \cite{KLPSS16}, require \emph{vertex sparsification}, that is, reducing the number of vertices in the graph.

In this work, the notion of vertex sparsification that we consider is \emph{cut sparsification}, introduced by \cite{HKNR98,Moitra09,LM10}. In this problem, we are given a graph $G$ and a set of \emph{terminals} $\T \subseteq V(G)$. Our goal is to construct a smaller graph $G'$ that approximates the value of all cuts in $G$ between terminals up to a multiplicative factor $q$, called the \emph{quality} of the sparsifier. Precisely, we wish to construct a graph $G'$ also containing the terminals $\T$ among its vertices, such that for any subset $S \subseteq \T$, the minimum cut in $G$ (respectively $G'$) between $S$ and its complement $\T \bs S$ differ by at most the multiplicative approximation $q$.

Desirable properties of an algorithm for producing vertex sparsifiers include achieving quality $q$ close to $1$, while still maintaining that the graph $G'$ constructed has few vertices, hopefully nearly linear in the number of terminals, which we denote as $k = |\T|$. Additionally, for applications we also want for the runtime of the vertex sparsification algorithm to be polynomial or even linear in the size of the original graph $G$. Previous results achieve various subsets of these properties. In the setting where the sparsifier $G'$ has vertex set exactly $\T$, and upper bound quality $q = O(\log k / \log\log k)$ was achieved \cite{Moitra09,LM10,CLLM10,EGKRTT14}, and a polynomial runtime was achieved by \cite{MM10}. In the setting where the quality $q = 1$, upper bounds of $2^{2^k}$ \cite{HKNR98,KR14} and lower bounds of $2^{\Omega(k)}$ \cite{KR13} were achieved. Additionally, \cite{KW12} achieved a bound of size $O(Z^3)$ for quality $q = 1$ in polynomial time for graphs with integer capacities, where $Z$ is the total degree of all terminals. Despite this, it is still not known whether quality $q = 1+\eps$ sparsifiers with $\O(\poly(k/\eps))$ vertices exist except in special cases \cite{AGK14,ADKKP16}.

In this paper we provide a vertex sparsification algorithm for $c$-edge connectivity, a thresholded version of cut sparsification, which maintains all cuts of size at most $c$ exactly, has size linear in the number of terminals and polynomial in $c$, and runs in polynomial time for all $c$ (\cref{thm:main}). This notion, which we call a \emph{connectivity-$c$ mimicking network}, was introduced in \cite{CDL21} to study dynamic connectivity problems and parametrized complexity, and has resulted in the first fully dynamic online algorithm for $c$-connectivities with almost constant update time for constant $c$ \cite{JS20}. Precisely, we say that graph $G'$ is a connectivity-$c$ mimicking network for $G$ with terminals $\T$ if all cuts between terminals in $G$ with at most $c$ edges are maintained exactly in $G'$.

The previous algorithm of \cite{CDL21} which builds a connectivity-$c$ mimicking network with $O(kc^4)$ edges does not run in polynomial time because it uses the idea of a well-linked decomposition from \cite{Chu12}, which requires an exact solution to a restricted version of the sparsest cut problem. As a result, the algorithm's runtime depended exponentially on $c$.

Our main contribution is to open up the matroid-based cut covering lemmas of \cite{KW12}, and combine this with a much weaker notion of a well-linked decomposition. Precisely, we avoid doing a full well-linked decomposition, and instead just to partition the graph $G$ into pieces whose ``expansion'' with respect to terminals is not too small, similar to an expander decomposition. Now our key lemma (\cref{lemma:expanderc}), which applies matroid theory and representative sets following \cite{KW12}, directly gives a bound on the size of a connectivity-$c$ mimicking network for these expander-like pieces in terms of the expansion. We then combine the connectivity-$c$ mimicking networks on the pieces to get a connectivity-$c$ mimicking network for the original graph. In this way, we can afford to use approximate sparsest cut algorithms (\cref{thm:sparsecut} \cite{ALN05}), which run in polynomial time. Unfortunately, we do not achieve a $O(m \poly(c))$ runtime due to requiring expensive linear algebra to find independent sets in matroids. However, we are optimistic that a $O(m \poly(c))$ runtime is achievable and that it has further applications to improved data structures for connectivity and flows, which we detail in \cref{subsec:future}.

\subsection{Our Results}
Our main result is that given a graph $G$ with $k$ terminals $\T$, we can build a connectivity-$c$ mimicking network with $\O(kc^3)$ edges in polynomial time.
\begin{theorem}
\label{thm:main}
Given any edge-capacitated graph $G$ with $n$ vertices along with a set $\T$ of $k$ terminals, there is an algorithm which constructs a connectivity-$c$ mimicking network $H$ of $G$ with $O(kc^3\log^{3/2}n\log\log n)$ edges in time $n^{O(1)}$.
\end{theorem}
Without needing a polynomial time algorithm, we can slightly improve the size. We remark that the mimicking networks we construct are all minors of $G$, hence Theorem 1.1 in \cite{CDL21} and the comment afterwards show that these mimicking networks can be constructed in $O(m(c\log n)^{O(c)})$ time, which is super-polynomial for $c \ge \log n$.
\begin{theorem}
\label{thm:exist}
Given any edge-capacitated graph $G$ with $n$ vertices along with a set $\T$ of $k$ terminals, there is a connectivity-$c$ mimicking network $H$ of $G$ with $O(kc^3)$ edges.
\end{theorem}
These results improve over Theorem 1.1 in \cite{CDL21} for any $c \ge \log n$, as the runtime there depended on $c^{O(c)}$, which is super-polynomial for $c \ge \log n$. This gives an improvement for the runtime for survivable network design on low treewidth graphs, following Theorem 6.3 of \cite{CDL21}, and we refer the readers to Section 6.2 of \cite{CDL21} for details.
\begin{corollary}
\label{cor:SNDP}
There is an algorithm that exactly solves {\sf Subset $c$-EC} on an input graph $G$ with $n$ vertices in time $n \exp\left(O(c^3\tw(G)\log(\tw(G)c)\right)$, where $\tw(G)$ denotes the treewidth of $G$.
\end{corollary}

\subsection{Related work}
Our result brings together ideas from several areas, including work on cut sparsification and mimicking networks, matroid theory and polynomial kernelization, and sparsest cut and expander decompositions.
\paragraph{Cut sparsification and mimicking networks.} Without additional vertices, the best known upper and lower bounds are $O(\log k / \log \log k)$ \cite{CLLM10,MM10} and $\Omega(\sqrt{\log k} / \log\log k)$ \cite{MM10}. \cite{Chu12} provides an algorithm constructing $O(1)$-quality sparsifier with $O(Z^3)$ edges in time $\poly(n) \cdot 2^Z$, and \cite{KW12} showed a polynomial time algorithm for a quality-$1$ sparsifier with $O(Z^3)$ edges using matroid theory. It is open in general whether $(1+\eps)$-quality cut sparsifiers with at most $\O(\poly(k/\eps))$ edges exist.

A different studied \emph{mimicking networks}, i.e. quality-$1$ cut sparsifiers, with size depending only on the number of terminals $k$. Here, upper bounds of $2^{2^k}$ \cite{HKNR98,KRTV12} and lower bounds of $2^{\Omega(k)}$ \cite{KR14} are known. On planar graphs, an upper bound of $O(k^22^{2k})$ \cite{KR13} and lower bound of $2^{\Omega(k)}$ \cite{KPZ17} are known. Additionally, \cite{CSWZ00} provides a $O(k \cdot 2^{2^{\tw(G)}})$ upper bound in bounded treewidth graphs, as well as several sharper results when the number of terminals is small.

\paragraph{Sparsest cut and expansion.} The conductance of a graph is a fundamental quantity that has been extensively studied. Cheeger's inequality \cite{Cheeger70} relates the conductance to the spectrum of the graph, and since there has been significant interest in efficiently approximating the conductance and the related sparsest cut problem. A $O(\log n)$ approximation was given by \cite{LR99}, and later approximation ratios of $O(\sqrt{\log n})$ for the uniform case \cite{ARV09} and $O(\sqrt{\log n}\log\log n)$ for the nonuniform case \cite{ALN05}. Additionally, there was later work towards making these algorithms more efficient \cite{AHK04,Sherman09}.

Related to the sparsest cut problem is the concept of \emph{expander decomposition}, that is, how to partition a graph so that all pieces are expanders? There has been significant work towards achieving linear time algorithms for expander decompositions \cite{ST04,NS17,SW19,CGL20}, with many works based on the cut-matching game \cite{KRV06}. Both the cut-matching game \cite{RST14} and expander decompositions \cite{KLOS14,CGP18,KPSW19,CDL21} have seen significant use in graph algorithms.

\paragraph{Polynomial kernelization and parametrization.} The concept of kernelization in parameterized algorithms is, given a parameter $k$ for a problem, to efficiently reduce the problem size to something depending only on $k$, while preserving necessary quantities. For example, if $k$ is the number of terminals in a graph $G$, one can ask whether the size of $G$ can be reduced to size polynomial in $k$ to maintain all terminals cuts exactly. \cite{KW12}, using tools from matroid theory and representative sets \cite{Lovasz77,Marx09}, shows polynomial kernels for several problems. Recently, \cite{W20} has built an improved size for kernelization of multiway-cut using a combination of graph partitioning and representative sets, somewhat similar to the algorithm of this work.

\subsection{Future directions}
\label{subsec:future}
There are several promising directions to pursue.

\paragraph{Improved algorithm runtime.} Further improving the algorithm to construct connectivity-$c$ mimicking networks of size $\O(kc^3)$ in $\O(m \cdot \poly(c))$ time would result in several potential applications. In particular, it may be possible to use such an algorithm to develop algorithms for fully dynamic online $c$-connectivity following work of \cite{JS20} that have update time $\O(\poly(c))$ per iteration. On the other hand, the per query time of \cite{JS20} depends exponentially on $c$ (at least $(c\log n)^c$). Additionally, there may be further extensions towards fully dynamic $(1+\eps)$-approximate maxflow algorithms in unweighted graphs.

One potential route to achieve a nearly linear construction runtime $\O(m \cdot \poly(c))$ is to reinterpret the concepts of representative sets on the direct sum matroids constructed in the proof of \cref{lemma:expanderc} combinatorially, in terms of matching, flows, and cuts. In this way, it may be possible to construct the representative set by running a maximum flow routine on the correct graph, instead of doing linear algebra over matroid representations.

\paragraph{Optimal $c$ dependence in sparsifier size.} It would be interesting to get an optimal dependence on $c$ in the size of the connectivity-$c$ mimicking network. We were only able to prove a lower bound of $2kc$, and a lower bound of $\Omega(kc)$ in the setting where all terminals have degree $1$. As such, we are inclined to believe an upper bound of $O(kc^2)$ or even $O(kc)$.

\paragraph{Application towards $(1+\eps)$-quality cut sparsification.} It would be exciting to apply the ideas of this work towards constructing $(1+\eps)$-quality cut sparsifiers with at most $\O(k/\poly(\eps))$ edges. A natural approach would be to pick a slightly super-constant $c = C\log n/\eps^2$, find a set of edges $F$ which covers all cuts of size at most $c$ between terminals, and apply a uniform sampling procedure \cite{BK96,ADKKP16} to the remaining edges. Na\"ively, this approach seems to not work due to the existence of several potential minimum cuts between a fixed terminal partition.
\section{Preliminaries}
\label{sec:prelim}

\subsection{Terminals and connectivity}
Throughout, we assume that we work with unweighted multigraphs. We may assume this by replacing an edge with weight $w$ with $w$ copies of edges with weight $1$. We work with multigraphs because our algorithms involve edge contractions, which naturally creates multigraphs.
Given a graph $G$ and disjoint sets $A, B \subseteq V(G)$ define the minimum cut between $A, B$ in $G$ as
\[ \mincut_G(A, B) \defeq \min \left\{|E_G(S, V(G)\bs S)| : A \subseteq S, B \subseteq V(G) \bs S \right\}. \]
Given this, we can formally define connectivity-$c$ mimicking networks.
\begin{definition}[$(\T, c)$-equivalence]
\label{def:tc}
Given graphs $G, H$, both containing terminals $\T$, we say that $G, H$ are $(\T, c)$-equivalent if for all $S \subseteq \T$ we have that
\[ \min\left(\mincut_G(S, \T \bs S), c\right) = \min\left(\mincut_H(S, \T \bs S), c\right). \]
If $G$ and $H$ are $(\T, c)$-equivalent, we can say that $H$ is a \emph{connectivity-$c$ mimicking network} for $G$.
\end{definition}
We would like to note that contracting edges cannot decrease any minimum cuts between subsets $A, B \subseteq V(G)$.

We will construct connectivity-$c$ mimicking networks for graphs $G$ by finding a subset $F \subseteq E(G)$ that covers all cuts of size at most $c$ in $G$ between terminals.
\begin{definition}[Cut covering]
\label{def:cutcover}
Let $G$ be a graph with terminals $\T$. We say that a set of edges $F \subseteq E(G)$ \emph{covers all $c$-cuts} if for all subsets $S \subseteq \T$ with $\mincut_G(S, \T\bs S) \le c$, there exists a subset of edges $F^S \subseteq F$ such that $F^S$ is a cut with $\mincut_G(S, \T\bs S)$ edges between $S$ and $\T \bs S$ in $G$.
\end{definition}
It is clear that if $F$ covers all $c$-cuts, then contracting $G$ onto $F$ gives a connectivity-$c$ mimicking network.

We may reduce to the situation where all terminals have degree $1$ by increasing the number of terminals by a factor of $c$. Given a graph $G$ with terminals $\T$, we construct a new graph $G^\new$ with terminals $\T^\new$ as follows. To construct $G^\new$, for each vertex $t \in \T$ add $c$ new vertices to $G$ and connect them to $t$, and let $\T^\new$ be the set of $|\T|c$ new vertices added. We can see that if a set of edges covers all $c$-cuts in $G^\new$ for terminals $\T^\new$, then it also does so for $G$ with terminals $\T$.

Given this reduction, we can show that if we partition the vertex set of $G$ and add terminals corresponding to the boundary edges of the partition, then it suffices to compute a set of edges covering all $c$-cuts on the partition pieces. Here, for a subset $X \subseteq V$, we define the boundary edges $\partial(X) \defeq E_G(X, V(G)\bs X)$, and $V[\partial(X)]$ to denote turning each boundary edge in $\partial(X)$ into a vertex of degree $1$, with the other endpoint lying in $X$. Here $\sqcup$ throughout denotes disjoint union.
\begin{lemma}[Partition of $G$]
\label{lemma:partition}
Let $G = (V, E)$ be a graph with terminals $\T$, all of degree $1$, and let $X = V \bs \T$. Let $X = X_1 \sqcup X_2 \sqcup \dots \sqcup X_p$ be a partition of $X$. For $1 \le i \le p$ define graph $G_i$ to have vertex set $V_i \defeq X_i \cup V[\partial(X_i)]$, terminals $\T_i \defeq V[\partial(X_i)]$, and edge set $E_i \defeq E(X_i) \cup \partial(X_i)$. Let $F_i$ cover all $c$-cuts in $G_i$. Then $\bigcup_i F_i$ covers all $c$-cuts in $G$.
\end{lemma}
\begin{proof}
For a subset $T \subseteq \T$, consider a minimum cut $(S, V \bs S)$ which separates $T$ and $\T \bs T$ in $G$ and is minimal. Say that this cut induces the terminal cut $(T_i, \T_i \bs T_i)$ on $X_i$. Note that this cut on $X_i$ clearly has at most $c$ edges, so $F_i$ covers this terminal cut. Because all terminals in all the $X_i$ (and $X$) correspond to boundary edges, we can combine the induced cuts on each $X_i$ to get a cut on $X$, as desired.
\end{proof}

\subsection{Matroids and representative sets}

Our algorithm will require several concepts from matroid theory and representative sets \cite{KW12}. We denote matroids $\M = (S, \I)$, where the ground set is $S$ with independent sets $\I$.

This work only considers representable matroids.
\begin{definition}[Representable matroid]
\label{def:repmatroid}
We say that a matroid $\M = (S, \I)$ is \emph{representable} if there is a field $\F$ and vectors $v_1, v_2, \dots, v_{|S|} \in \F^{|S|}$ such that a set $X \in \I$ if and only if the vectors $(v_x)_{x \in X}$ are independent over $\F$.
\end{definition}
The \emph{rank} of a matroid is the maximum size of any independent set. Given matroids $\M_1 = (S_1, \I_1), \M_2 = (S_2, \I_2)$, the direct sum $\M = \M_1 \oplus \M_2$ is defined to have ground set $S_1 \sqcup S_2$ and independent sets $X_1 \sqcup X_2$ for $X_1 \in \I_1, X_2 \in \I_2$. The direct sum of representable matroids is representable.

\emph{Transversal matroids} and \emph{gammoids} are examples of representable matroids.
\begin{definition}[Transversal matroid]
\label{def:transmatroid}
Given a bipartite graph $G = (A, B)$, a \emph{transversal matroid} $\M = (B, \I)$ is defined so that a set $X \subseteq B$ is independent if and only if there is a matching in $G$ that covers $X$.
\end{definition}
Gammoids are the dual of transversal matroids.
\begin{definition}[Gammoids]
\label{def:gammoid}
For a directed graph $G = (V, E)$ and vertex subset $\T \subseteq V$, a \emph{gammoid} $\M = (V, \I)$ is defined so that a set $X \subseteq V$ is independent if and only if there are $S$ vertex disjoint paths from $\T$ to $X$.
\end{definition}
Combining Proposition 3.6 and 3.11 of \cite{Marx09} shows that a representation of a gammoid can be constructed in randomized polynomial time with failure probability exponentially small in $n$.
\begin{theorem}[\!\!\cite{Marx09}]
\label{thm:gammoidrep}
For a directed graph $G = (V, E)$ with $n$ vertices and vertex subset $\T \subseteq V$, the gammoid defined in \cref{def:gammoid} is a matroid, and a representation can be found in polynomial time with probability at least $1-2^{-n}$.
\end{theorem}
We can view undirected graphs as directed graphs by turning each undirected edge into two directed edges in opposite directions.

We require a specific version of a fundamental lemma about finding small representative sets due to \lovasz~\cite{Lovasz77} and Marx \cite{Marx09}.
\begin{definition}[Representative sets]
\label{def:repset}
Given a representable matroid $\M = (S, \I)$ and a collections $\J$ of subsets of $S$, we say that $\J^* \subseteq \J$ is a \emph{representative} for $\J$ if the following holds: for every set $Y \subseteq S$, if there is a set $X \in \J$ disjoint from $Y$ such that $X \cup Y \in \I$, then there is a set $X^* \in \J^*$ disjoint from $Y$ such that $X^* \cup Y \in \I$.
\end{definition}
\begin{lemma}[\!\!\cite{Lovasz77,Marx09}]
\label{lemma:repset}
Let $\M_i = (S_i, \I_i)$ be representable matroids over a field $\F$ for $1 \le i \le p$, and let $\M = \M_1 \oplus \M_2 \oplus \dots \oplus \M_p$ be their direct sum. Define \[ S_1 \times S_2 \times \dots \times S_p = \left\{\{x_1,x_2,\dots,x_p\} : x_i \in S_i\right\}, \] that is, the collection of subsets of $S$ that have exactly one element from each $S_i$. Then every subset $\J \subseteq (S_1 \times S_2 \times \dots \times S_p)$ has a representative set $\J^*$ of size at most $\prod_{i=1}^p \rank(\M_i)$, and $\J^*$ is computable in time $|\J|^{O(1)}$.
\end{lemma}
For completeness, we provide a proof in \cref{proofs:repset}.

\subsection{Sparsest cut and conductance}
We will require polynomial time algorithms for the nonuniform sparsest cut problem. The best known approximation ratio is $O(\sqrt{\log n}\log\log n)$.
\begin{definition}[Nonuniform sparsest cut]
\label{def:sparsecut}
Given graphs $G, H$ on the same vertex set $V$, define the sparsest cut of $G$ with respect to $H$ as
\[ \sc(G, H) \defeq \min_{S \subseteq V} \frac{|E_G(S, V \bs S)|}{|E_H(S, V \bs S)|}. \]
\end{definition}
We will let $\sc_n$ denote the approximation ratio of the best known polynomial time approximation algorithm for nonuniform sparsest cut.
Given the result below, $\sc_n = O(\sqrt{\log n} \log \log n)$.
\begin{theorem}[Sparsest cut approximation (\!\!\cite{ALN05} Theorem 1.2)]
\label{thm:sparsecut}
Given graphs $G, H$ on the same vertex set $V$, there is a polynomial time algorithm which produces a set $S$ satisfying
\[ \frac{|E_G(S, V \bs S)|}{|E_H(S, V \bs S)|} \le O(\sqrt{\log n}\log\log n) \sc(G, H). \]
\end{theorem}

The sparsest cut problem (up to constants) captures concepts such as expansion and conductance. A useful notion we will use is expansion among terminals in a graph $G$.
\begin{definition}[Terminal expansion]
\label{def:terminalexpander}
Given a graph $G$ with terminals $\T$, we say that $G$ is a \emph{$\phi$-terminal expander} if
\[ \frac{|E_G(S, V \bs S)|}{\min(|S \cap \T|, |(V \bs S) \cap \T|)} \ge \phi \forall S \subseteq V. \]
\end{definition}
Indeed, if we take $H = K_{\T}$, i.e. the complete graph over terminals $\T$ (with additional isolated vertices added so that $V(H) = V(G)$) the terminal expansion is within a factor of $2$ of $|\T|\sc (G, H)$.
\section{Polynomial Time Vertex Sparsification}
\label{sec:polytime}

In this section we will show \cref{thm:main,thm:exist}. We first give a high-level overview of the proof, and then describe the details.
\subsection{Overview of proof}
Let us recall the proof of \cite{CDL21}, which used the ideas of well-linkedness \cite{Chu12} and results of \cite{KW12}. For a graph $G = (V, E)$ with terminals $\T$, \cite{CDL21} computed a \emph{well-linked partition} of the graph, i.e. partitioned the set $X = V \bs \T$ into $X = X_1 \sqcup X_2 \sqcup \dots \sqcup X_p$ so that each $X_i$ either had at most $O(c)$ boundary edges (corresponding to terminals), or were connectivity-$c$ well-linked. Here, connectivity-$c$ well-linked (which we refer to as simply ``well-linked'' going forwards) was defined to mean that every cut in $X_i$ separating terminals had at least as many edges as the number of terminals of the smaller side. This guarantees that if $X_i$ is further partitioned along a cut that isn't well-linked, the number of terminals on both sides decreases. To finish, \cite{CDL21} showed that every well-linked set can be contracted to a single vertex, and for the remaining pieces $X_i$ with at most $O(c)$ boundary edges, directly cited the result of \cite{KW12} to construct a connectivity-$c$ mimicking network on each piece of size $O(c^3)$.

This approach does not lead to an algorithm which runs in time polynomial in $c$, as deciding whether a partition piece $X_i$ is well-linked is an instance of nonuniform sparsest cut (\cref{def:sparsecut}). Therefore, even if there is ever a piece $X_i$ that isn't well-linked, finding a cut which certifies that $X_i$ isn't well-linked (and thus makes progress) may require time exponential in $c$. However, if we use \cref{thm:sparsecut} to find an approximate sparse cut to partition and recurse along, the number of edges in the cut may be larger than the number of terminals on the smaller side, and thus after partitioning, the number of terminals increases in the recursive subproblem, preventing progress.

To get around this issue, our key observation is that we can open up the matroid and representative set tools from \cite{KW12} to prove in \cref{lemma:expanderc} a significantly improved bound on the size of a connectivity-$c$ mimicking network for any graph that is a $\phi$-terminal expander, i.e. the small side of any terminal cut of size at most $c$ has at most $c\phi^{-1}$ edges. In this way, if we choose $\phi = \frac{1}{C\log^2 n}$ say, we can certify that partition piece $X_i$ is a $\phi$-terminal expander or find a cut to make significant recursive progress.

\subsection{Proof of \cref{thm:main,thm:exist}}
We show our key lemma that if we have a bound on number of terminals in the smaller side of every cut with at most $c$ edges, then we immediately get a bound on the size of the a connectivity-$c$ mimicking network.
\begin{lemma}
\label{lemma:expanderc}
Let $G = (V, E)$ be a graph with $k$ terminals $\T$, each with degree $1$. Let $d \ge c$ be a parameter such that for every subset $S \subseteq V$ with $|E_G(S, V \bs S)| \le c$, we have $\min(|S \cap \T|, |(V \bs S) \cap \T|) \le d$. Then there is a set $F$ of at most $O(kcd)$ edges which covers all $c$-cuts in $G$ with terminals $\T$ which is computable in $n^{O(1)}$ time.
\end{lemma}
To show this, we first reduce to the setting of vertex cuts by turning each edge in $G$ into a vertex, as gammoids are defined in terms of vertex cuts. Call the new graph $G_\split$. In the setting of vertex cuts, we say that a vertex $v$ is \emph{essential} if there exists a partition of the terminals $\T = A \sqcup B$ such that $v$ is in all minimum vertex cuts between $A, B$ in $G$. If there is some nonessential vertex $v$, then the edge corresponding to $v$ in the original graph $G$ may be contracted without affecting any minimum cuts. We formalize these definitions and observations in the following claims.
\begin{definition}[Vertex cuts]
Given a directed graph $G$ and disjoint sets $A, B \subseteq V$, the size of the minimum vertex cut between $A$ and $B$ is the smallest size of a set $C \subseteq V$ such that there are no paths between $A$ and $B$ in $G \bs C$. In particular, $C$ \emph{may intersect} $A \cup B$.
\end{definition}
By Menger's Theorem, the minimum vertex cut between $A$ and $B$ is the maximum number of vertex disjoint flow paths between $A$ and $B$ in $G$.
\begin{definition}[Essential vertices]
\label{def:essential}
Given a directed graph $G$, terminals $\T$, and a connectivity parameter $c$, we say that a vertex $v$ is \emph{essential} if there exists a partition $\T = A \sqcup B$ such that the minimum vertex cut between $A$ and $B$ has size at most $c$ and $v$ is involved in all minimum vertex cuts between $A$ and $B$ in $G$.
\end{definition}
\begin{lemma}
\label{lemma:augment}
Consider a directed graph $G$, terminals $\T$, connectivity parameter $c$, partition $\T = A \sqcup B$, and a vertex $v \in V$.
Let $v'$ be a sink-only copy of $v$, i.e., $v'$ has in-edges from all in-neighbors of $v$, and let $v''$ be a source-only copy of $v$, i.e., $v''$ has out-edges edges to all out-neighbors of $v$.
If $v$ is essential to the cut $(A, B)$, then there exists a minimum vertex cut $C \subseteq V$ with at most $c$ vertices and $v \in C$ such that there are $|C|+1$ vertex disjoint paths from $A$ to $C \cup v'$ and $|C|+1$ vertex disjoint paths from $C \cup v''$ to $B$.
\end{lemma}
\begin{proof}
Assume for contradiction that there are less than $|C|+1$ vertex disjoint paths from $A$ to $C \cup v'$. Then there is a minimal set $C'$ with $|C'| = C$ such that $A$ is separated from $C \cup v'$ in $V(G) \bs C'$. If $v' \in C'$, then clearly $C' \bs v'$ separates $A$ from $C$, contradicting that the minimum cut between $A$ and $C$ is size $|C|$. So we may assume that $v' \notin C'$. If $v \notin C'$, $C'$ separates $A$ from $C$ and $v \notin C'$, so $v$ is not essential. So we may assume that $v' \notin C'$ and $v \in C'$.

We claim that there is a path from $A$ to $v$ in $G \bs (C' \bs v)$. Indeed, $(C' \bs v)$ cannot separate $A$ from $C$ as there are $|C|$ vertex disjoint paths from $A$ to $C$, but $C'$ does separate $A$ from $C$. Hence, there must be a path from $A$ to $v$ in $G \bs (C' \bs v)$. As $v'$ is a sink-only copy of $v$, there must be a path from $A$ to $v'$ in $G \bs C'$, as desired. The symmetric argument applies to $B$ and $C \cup v''$.
\end{proof}

To bound the number of essential vertices, we apply \cref{lemma:repset} for carefully constructing matroids $\M_1, \M_2, \M_3$. This suffices, as we can pick any single nonessential vertex to contract in the original graph, and then repeat the argument on the contracted graph with one less edge.
\begin{proof}[Proof of \cref{lemma:expanderc}]
We first explain a reduction to vertex cuts.
\paragraph{Reduction to vertex cuts.} To do this, turn each edge in $G$ into a vertex, by putting a vertex in the middle of the edge. We call these \emph{split vertices}. For each non-terminal vertex $v \in V(G)$, replace it with a complete graph with $2c$ vertices, and connect each of these $2c$ vertices to the corresponding split vertices of the neighboring edges. Call this new graph $G_\split$. We leave the terminals untouched. Now, for a terminal partition $\T = A \sqcup B$, note that for a vertex cut of size at most $c$ between $A$ and $B$ in $G_\split$ will not involve any vertices in the complete graphs of size $2c$, as there is no reason to contain one of these vertices and not the rest. In particular, it is impossible to remove a proper subset of a clique of size $2c$ and disconnect any pair of neighbors, because all vertices in the clique have the same neighborhood. Additionally, vertex cuts that involve terminal vertices in $G_\split$ correspond to picking the edges adjacent to terminals in $G$, as each terminal in $\T$ has degree $1$. In this way, there is a bijection between edge cuts with at most $c$ edges between terminals in $G$ and vertex cuts with at most $c$ vertices between terminals in $G_\split$.

Additionally, note that if some split vertex $v$ is nonessential in $G_\split$ for terminals $\T$ and connectivity $c$, then we may contract its corresponding edge in $G$ and maintain all terminal cuts with size at most $c$ in $G$. In this way, it suffices to upper bound the number of essential vertices in $G_\split$. We do this by using representative sets.

\paragraph{Building the direct sum matroid.} As in \cref{lemma:augment}, create a graph $G_\split'$ which is a copy of $G$ with additional sink-only copies $v'$ for each non-terminal vertex $v \in V(G_\split)$, and $G_\split''$ with source-only copies $v''$. Define matroids $\M_1 = (V, \binom{V}{\le c})$ with ground set $V$, and independent sets are all sets of size at most $c$. Define $\M_2$ to be the gammoid with vertex subset $\T$ on $G_\split'$, and let $\M_3$ be the gammoid with vertex subset $\T$ on $G_\split''$, except restrict all independent sets down to size $c+d$. Define $\M = \M_1 \oplus \M_2 \oplus \M_3$. As $\M_1$ is trivially representable, and $\M_2, \M_3$ are representable by \cref{thm:gammoidrep}, we have that $\M$ is representable.

\paragraph{Bounding essential vertices via representative sets.} In the setting of \cref{lemma:repset}, define $\J = \{(v, v', v'') : v \in V \},$ where $v'$ and $v''$ are the sink-only and source-only copies of $v$. Let $\J^*$ be a representative set for $\J$. We show that any essential vertex $v$ must correspond to a set in $\J^*$. Indeed, consider an essential vertex $v$ involved in a minimum vertex cut $C$ between $A$ and $B$, where $\T = A \sqcup B$. We may assume that $|B| \le d$ by the problem condition. Consider the independent set $Y = (C \bs v, A \cup C, B \cup C)$ in $\M$, where $A \cup C$ is independent in $\M_2$ because there are $|C|$ vertex disjoint flow paths from $B$ to $C$ by the maxflow-mincut theorem, and similarly for $B \cup C$ in $\M_3$, as $|B \cup C| \le c+d$. By \cref{lemma:augment}, we have that for $X = (v, v', v'')$ that $X \sqcup Y$ is independent. Thus, there exists a set $X^* = (u, u', u'') \in \J^*$ such that $X^* \sqcup Y$ is independent. Note that if $A \cup C \cup u'$ is independent, there must be $|C|+1$ paths from $B$ to $C \cup u'$, hence $u$ must be on the same side of the cut as $B$ in $G_\split \bs C$. If $B \cup C \cup u''$ is independent, then $u$ must be on the same side as $A$, hence $u \in C$. But $u$ must also be disjoint from $C \bs v$, hence $u = v$. So all essential vertices $v$ must correspond to triples in $\J^*$.

\paragraph{Size and runtime bound.} As $|\J^*| \le \rank(\M_1)\rank(\M_2)\rank(\M_3) \le ck(c+d) \le O(kcd)$, there are at most $O(kcd)$ essential vertices $v$. At the end, all uncontracted split vertices / edges of the original graph $G$ cover all $c$-cuts by construction. The runtime is polynomial as all representations of $\M_1, \M_2, \M_3, \M$ can be computed in polynomial time, and $\J^*$ is computable in polynomial time by \cref{lemma:repset}.
\end{proof}

We can now prove \cref{thm:exist}, using \cref{algo:exist}. The algorithm maintains a partition $X = X_1 \sqcup \dots \sqcup X_p$, and refines the partition until a stopping condition. For every piece $X_i = (V_i, E_i)$, we define its vertex, edge, and terminal set $\T_i$ as in \cref{lemma:partition}. Precisely, its terminal set consists of the boundary edges in $\partial(X_i)$ and the edges are those in $E(X_i)$ plus boundary edges.
\begin{algorithm2e}[h]
\caption{\label{algo:exist} Given a graph $G = (V, E)$ with all $k$ terminals $\T$ with degree $1$, returns a set $F$ of $O(kc^2)$ edges which covers all $c$-cuts in $G$ with terminals $\T$.}
\SetKwProg{Proc}{procedure}{}{}
\Proc{\Existence$(G, \T)$}{
	Initialize $X \assign V \bs \T$, and maintain a partition $X = X_1 \sqcup \dots \sqcup X_{p'}$ at all times. \Comment{Notationally, we let $X_i$ has vertex set $V_i$, edges $E_i$, and terminals $\T_i$.} \\
	\While{There is $i \in [p']$ and $S \subseteq V_i$ such that $|E_{X_i}(S, V_i \bs S)| \le c$ and $\min(|S \cap \T_i|, |(V_i \bs S) \cap \T_i|) \ge 3c$}{
		Replace $X_i$ in the partition with $X_i \cap S$ and $X_i \cap (V_i \bs S)$. \\
	}
	For $1 \le i \le p$ apply \cref{lemma:expanderc} to find a set $F_i$ of edges which covers all $c$-cuts in $X_i$ with terminals $\T_i$. \\
	\Return $F \defeq \bigcup_{i=1}^p F_i$.
}
\end{algorithm2e}

\begin{proof}[Proof of \cref{thm:exist}]
Given a graph $G$ with all $k$ terminals $\T$ with degree $1$, we show that \cref{algo:exist} returns a set $F$ of $O(kc^2)$ edges which covers all $c$-cuts in $G$ with terminals $\T$, which we can contract onto to get a connectivity-$c$ mimicking network.
This shows \cref{thm:exist}, as recall that in \cref{sec:prelim} we showed that we can reduce to the case where all terminals have degree $1$ by increasing the number of terminals by a factor of $c$.

The correctness of the algorithm follows from \cref{lemma:partition}, specifically that it suffices to compute a partition of $X$ and then take a union of the edges $F_i$ which cover all $c$-cuts on each $X_i$. It remains to bound the size of $F$.

We first bound the number of pieces $p$ by setting up a potential function over the partition $\psi(X_1, X_2, \dots, X_{p'}) \defeq \sum_{i=1}^{p'} (|\T_i|-3c)$. Note that at the beginning $\psi(X) = k-3c$. Note that when $X_i$ is split into $A = X_i \cap S$ and $B = X_i \cap (V_i \bs S)$, $A$ now has $|S \cap \T_i| + |E_{X_i}(S, V_i \bs S)| \ge 3c$ terminals, and $B$ has $|(V \bs S) \cap \T_i| + |E_{X_i}(S, V_i \bs S)| \ge 3c$ terminals by the cutting condition in the while statement of \cref{algo:exist}. Hence $\psi \ge 0$ always. Additionally, note that when $X_i$ is split into $A$ and $B$, $\psi$ decreases by
\begin{align*}
&(|\T_i|-3c) - (|S \cap \T_i| + |E_{X_i}(S, V_i \bs S)|-3c) - (|(V \bs S) \cap \T_i| + |E_{X_i}(S, V_i \bs S)| - 3c) \\
&= 3c - 2|E_{X_i}(S, V_i \bs S)| \ge c.
\end{align*}
Hence there can at most $\psi(X)/c$ splits, so $p \le k/c$. By this computation, each split increases the number of terminals by at most $2c$, hence at the end we have $\sum_{i=1}^p |\T_i| \le k + pc \le 3k$.

To finish, note that by \cref{lemma:expanderc} for $d = 3c$ we get that
\[ |F| \le \sum_{i=1}^p |F_i| \le \sum_{i=1}^p O(|\T_i|cd) = O(kc^2) \] as desired.
\end{proof}

We can change the cutting condition in the above proof and apply a sparsest cut approximation algorithm to show \cref{thm:main}.
\begin{algorithm2e}[h]
\caption{\label{algo:polytime} Given a graph $G = (V, E)$ with all $k$ terminals $\T$ with degree $1$, returns a set $F$ of $O(kc^2\sc_n\log n)$ edges which covers all $c$-cuts in $G$ with terminals $\T$ in polynomial time.}
\SetKwProg{Proc}{procedure}{}{}
\Proc{\PolyTimeNetwork$(G, \T)$}{
	Initialize $X \assign V \bs \T$, and maintain a partition $X = X_1 \sqcup \dots \sqcup X_{p'}$ at all times. \Comment{Notationally, we let $X_i$ has vertex set $V_i$, edges $E_i$, and terminals $\T_i$.} \\
	$\phi \assign \frac{1}{10\sc_n\log n}$. \Comment{Expansion parameter.} \\
	\While{There is $i \in [p']$ such that $X_i$ is not a $\phi$-terminal expander}{
		Compute $S \subseteq V_i$ such that $\frac{|E_{X_i}(S, V_i \bs S)|}{\min(|S \cap \T_i|, |(V_i \bs S) \cap \T_i|)} \le \sc_n \phi$ using \cref{thm:sparsecut}. \label{line:sparsecut} \\
		Replace $X_i$ in the partition with $X_i \cap S$ and $X_i \cap (V_i \bs S)$. \\
	}
	For $1 \le i \le p$ apply \cref{lemma:expanderc} to find a set $F_i$ of edges which covers all $c$-cuts in $X_i$ with terminals $\T_i$. \label{line:matroid} \\
	\Return $F \defeq \bigcup_{i=1}^p F_i$.
}
\end{algorithm2e}

\begin{proof}[Proof of \cref{thm:main}]
Given a graph $G$ with all $k$ terminals $\T$ with degree $1$, we show that \cref{algo:polytime} returns a set $F$ of $O(kc^2\sc_n\log n)$ edges which covers all $c$-cuts in $G$ with terminals $\T$, which we can contract onto to get a connectivity-$c$ mimicking network.
This shows \cref{thm:main}, as recall that in \cref{sec:prelim} we showed that we can reduce to the case where all terminals have degree $1$ by increasing the number of terminals by a factor of $c$.

Correctness follows from \cref{lemma:partition} along with \cref{thm:sparsecut}, specifically that there is a polynomial time algorithm to compute the set $S$ desired in line \ref{line:sparsecut} of \cref{algo:polytime} assuming that $X_i$ is not a $\phi$-terminal expander. As both line \ref{line:sparsecut} and line \ref{line:matroid} run in polynomial time, the whole algorithm runs in polynomial time.

We now bound the size of $F$, by first bounding $\sum_{i=1}^p |\T_i| = O(|\T|)$. To show this, we use a standard argument in bounding the number of edges cut in an expander decomposition. Consider a piece $X'$ with terminals $\T'$ in a partition of $X$, and consider how it gets partitioned further. Imagine running the partitioning procedure in \cref{algo:polytime} on $X'$ until each piece has at most $2|\T'|/3$ terminals. Let $t_1, t_2, \dots, t_j$ denote the number of terminals on the smaller side in each step of the partition. Note that the number of terminals in the larger piece decreases by at least $4t_i/5$ during step $i$, as $\frac{1}{5\log n}t_i \le t_i/5$. Thus, we know that $\sum_i t_i \le 5|\T'|/4.$ During the partition, the number of new terminals added is at most $\sum_{i=1}^j 2\phi\sc_n t_i \le \frac{1}{5\log n} \cdot 5|\T'|/4 = \frac{1}{4\log n}|\T'|$ by the cutting condition. Therefore, at the bottom level, the total number of terminals is at most $\left(1 + \frac{1}{4\log n}\right)^{\log_{3/2}n}|\T| = O(|\T|)$.

Because each $X_i$ at the end is a $\phi$-terminal expander, we may take $d = \phi^{-1}c$ in \cref{lemma:expanderc}, so we get a final bound of
\[ |F| \le \sum_{i=1}^p O(|\T_i|cd) = O(kc^2\phi^{-1}) = O(kc^2 \sc_n\log n). \]
\end{proof}

\section*{Acknowledgments}
The author would like to thank Yunbum Kook for feedback on an earlier version of this manuscript, and Richard Peng for useful discussions and encouragement.

{\small
\bibliographystyle{alpha}
\bibliography{refs}}

\appendix

\section{Missing Proofs}

\subsection{Proof of \cref{lemma:repset}}
\label{proofs:repset}
\begin{proof}
Let $r_i = \rank(\M_i)$, and let $v_x^i \in \F^{r_i}$ for $x \in S_i$ be a representation for $\M_i$ over $\F^{r_i}$. Define the tensor product vector space $K = \F^{r_1} \otimes \F^{r_2} \otimes \dots \otimes \F^{r_p}$. For a set $\{x_1,\dots,x_p\} = X \in \J$, define $w(X) = v_{x_1}^1 \otimes v_{x_2}^2 \otimes \dots \otimes v_{x_p}^p \in K$. Now, let $\J^*$ be sets $X$ that form a maximal independent set of $\{w(X)\}_{X \in \J}$. Therefore, $|\J^*| \le \dim(K) = \prod_{i=1}^p r_i = \prod_{i=1}^p \rank(\M_i).$

It suffices to argue that $\J^*$ is a representative set for $\J$, containing sets $X_1, X_2, \dots, X_{|\J^*|}$. Consider the tensor product exterior algebra \[ \bigwedge \defeq \bigwedge(\F^{r_1}) \otimes \bigwedge(\F^{r_2}) \otimes \dots \otimes \bigwedge(\F^{r_p}), \] where the bilinear wedge product is defined component wise. Now, for an independent set $\{Y_1, \dots, Y_p \} = Y \in \I$, where each $Y_i \in \I_i$, define
\[ w(Y) \defeq (\wedge_{y \in Y_1} v_y^1) \otimes \dots \otimes (\wedge_{y \in Y_p} v_y^p). \]
Given this, note that $X \sqcup Y \in \I$ if and only if $w(X) \wedge w(Y) \neq 0$. Because $\J^*$ forms a maximal independent set, there are constants $c_i \in \F$ such that $w(X) = \sum_{i=1}^{|\J^*|} c_i w(X_i)$. Therefore, $w(X) \wedge w(Y) = \sum_{i=1}^{|\J^*|} c_i (w(X_i) \wedge w(Y))$, so there exists an $i$ such that $w(X_i) \wedge w(Y) \neq 0$, as desired.

To analyze the runtime, we may assume that $|\J| \ge \prod_{i=1}^p r_i = \dim(K)$, or we are done. The only step is to perform linear algebra over $K$, which costs $(\dim(K) + |\J|)^{O(1)} = |\J|^{O(1)}$ time, by the condition on $|\J|$.
\end{proof}

\end{document}